\documentclass[12pt]{amsart}
\usepackage{fullpage}
 
\usepackage{enumerate}

\usepackage{graphicx}

\usepackage{amsfonts}
\usepackage{amssymb}
\usepackage{amsmath}
\usepackage{amsthm}
\usepackage{latexsym}

\newtheorem{theorem}{Theorem}[section]
\newtheorem{lemma}[theorem]{Lemma}
\newtheorem{corollary}[theorem]{Corollary}
\newtheorem{proposition}[theorem]{Proposition}
\newtheorem{definition}[theorem]{Definition}

\theoremstyle{plain}
\theoremstyle{plain}

\theoremstyle{remark}
    \newtheorem{remark}[theorem]{Remark}

\newcommand{\EE}{\mathbb{E} }
\newcommand{\1}{\mathbf{1} }
\newcommand{\PP}{\mathbb{P} }
\newcommand{\QQ}{\mathbb{Q} }
\newcommand{\RR}{\mathbb{R} }

\newcommand{\ff}{\mathcal{F} }

\newcommand{\ignore}[1]{}

\begin{document}

\title{Polynomial term structure models}
\author{Si Cheng \and Michael R. Tehranchi
 \\
University of Cambridge }
\address{Statistical Laboratory\\
Centre for Mathematical Sciences\\
Wilberforce Road\\
Cambridge CB3 0WB\\
UK}
\email{sc591@cam.ac.uk, \ m.tehranchi@statslab.cam.ac.uk  }

\date{\today}
 
\maketitle

\begin{abstract}
In this article, we explore a class of tractable interest rate models that have the property
that the price  of a zero-coupon bond  can be expressed as a polynomial  of a state diffusion
process.  
Our results include a classification of all such time-homogeneous single-factor
models in the spirit of Filipovic's
maximal degree theorem  for exponential polynomial models, as well as an
explicit characterisation of the set of feasible parameters in the case when
the factor process is bounded.  Extensions to time-inhomogeneous and multi-factor polynomial
models are also considered.
\end{abstract}

\section{Introduction} 

Given an integer $d \ge 1$ and a non-empty open subset  $I \subseteq \RR^d$, a $d$-factor arbitrage-free model of the risk-free interest rate term
structure can be built from  four functions, $R:I \to \RR$, $G: \RR_+ \times I \to \RR$, $b: I \to \RR^d$ and $a: I  \to \RR^{d \times d}$,
satisfying the following hypotheses:

\vskip .25cm

\noindent \textbf{Hypothesis (PDE)}.
 The function $G$ is twice-continuously differentiable and satisfies
the partial differential equation 
$$  
\partial_{\tau} G = \sum_{1 \le i \le d} b_i \ \partial_{z_i} G + \frac{1}{2} \sum_{1 \le i,j \le d}
a_{ij}  \ \partial_{z_iz_j} G - R G  \mbox{ on } \RR_+ \times I,
$$ 
 with boundary condition
$$
G(0,z )  = 1 \mbox{ for all } z \in I;
$$  

\vskip .25cm

\noindent \textbf{Hypothesis (SDE)}. There exists a function $\sigma: I \to \RR^{d \times m}$ such that
$a = \sigma \sigma^\top$ and  such that for 
all $z \in I$ the stochastic differential equation
$$
dZ_t =  b(Z_t) dt + \sigma(Z_t) dW,\ \ Z_0 = z
$$
has a non-explosive weak solution $(\Omega, \ff, \QQ; Z, W)$ such that the process $Z$ 
 takes values in $I$ and  where $W$ is an $\RR^m$-valued Brownian motion.

\vskip .25cm

Indeed, given such functions $R,G, b$ and $ a$ satisfying the above hypotheses, one
need only fix $z \in I$, 
and let $Z$ be a solution of the stochastic differential equation with $Z_0=z$, where $Z_t$ models 
the time-$t$ value of the economic factor.  The  time-$t$ spot interest rate is then modelled as 
$$
r_t = R( Z_t)
$$
and the time-$t$ price of a zero-coupon bond of maturity $T$ is modelled as
$$
P_{t,T} =G(T-t, Z_t).
$$
Note that $P_{T,T} = G(0,Z_T) = 1$ by the boundary condition  
and by It\^o's formula and the  partial differential equation   the
discounted bond prices  $e^{-\int_0^t r_s ds} P_{t,T},
$
are local martingales for all $T \ge 0$.  In particular, the  measure
$\QQ$ is a local martingale measure for the model, and hence there is no arbitrage
in the bond market.

One usually takes the functions $R$, $b$ and $a$ as given, and then solves
the  partial differential equation for $G$.  In practice, such an equation could be solved numerically.
 However, in this paper, we turn things around and assume that the function $G$ takes a specific form.  

The motivation for this study comes from the problem of calibrating the model.
Indeed, a practitioner is actually interested in a family of functions $(R_{\theta},
G_{\theta}, b_{\theta}, a_{\theta})_{\theta \in \Theta}$ where $\theta$ is
an unknown parameter or vector of parameters.  Given a collection  of observed initial
bond prices $P_{0, T}$ for various maturites $T \in \mathcal T$, one then
tries to find $\theta$ to minimise  some notion of distance   between 
the observed prices $(P_{0,T})_{T \in \mathcal T}$ and the predicted prices $(G_{\theta}(T, z) )_{T \in \mathcal T}$.  
It is generallly computationally expensive to solve the  partial differential equation  numerically
to generate the values of $G_{\theta}(T,z)$ for all, or at least a representative
sample of,  $T \in \mathcal T$ and $\theta \in \Theta$.  
Therefore, there has been continuing
interest in developing tractable models, where the function $G_{\theta}$ is of a reasonably explicit
form.

Perhaps the two most famous tractable factor models are those of Vasicek \cite{V} and Cox, Ingersoll \& Ross \cite{CIR}.
In these models the factor is scalar and identified with the spot interest rate, so in the notation above, $d=1$
and $R(z)=z$, while
the functions $b$ and $a$ are  affine  and the function $G$ is of the exponential
affine form
$$
G(\tau,z) = e^{h_0(\tau)  + h_1(\tau) z }.
$$
In the case of exponential affine models, it is well-known that the partial differential equation  reduces to a system of coupled Riccati
ordinary differential equations for the functions $h_0$ and $h_1$ 
and the boundary condition  becomes $h_0(0)=h_1(0) = 0$.
Furthermore,   the  corresponding stochastic
differential equation   always has a unique \textit{local} solution. While the
local solution to the Vasicek stochastic differential equation
 is in fact the unique global solution, the situation with the
Cox--Ingersoll--Ross stochastic differential
 is more delicate: for some values of the parameters, local solutions  
may explode in finite time by hitting the boundary of the state space. 
 Duffie \& Kan \cite{DK} studied exponential affine models where the factor process is of arbitrary dimension $d \ge 1$,
finding conditions under which 
the corresponding stochastic differential equation  has a
non-explosive solution. Subsequently, there has been 
 a considerable body of research on the properties of these exponential affine 
models.  A notable contribution to this literature
is a general characterisation of exponential affine term
structure models by Duffie,  Filipovi\'c \& Schachermayer \cite{DFS}. 

An exponential affine model can be considered a special case of the family of exponential
quadratic models.  An early example of a quadratic model was proposed by Longstaff \cite{L},
and has since been developed and generalised by Jamshidian \cite{J}, Leippold \&  Wu \cite{LW},
and  Chen,  Filipovi\'c \& Poor \cite{CFP} among others.

One may wonder if there exist non-trivial exponential cubic (or higher degree) models.
  Filipovi\'c answered this
question in the negative, by showing that the maximal degree for exponential polynomial models
is necessarily two.  That is to say, the exponential quadratic models are indeed
the most general class of exponential polynomial models.

In this article, we consider a related   class of bond pricing functions, in
 which the function $G(\tau, \cdot)$ itself is a polynomial.  We introduce
the following hypothesis:

\vskip .25cm

\noindent \textbf{Hypothesis (POLY)} There exists an
integer $n \ge 1$ such that the function $G$ is of the form
$$
G(\tau, z) = \sum_{k_1 + \ldots +k_d \le n} g_k(\tau) z^k \mbox{ for all } (\tau, z) \in \RR_+ \times I,
$$
where for $k = (k_1, \ldots , k_d) \in \mathbb{Z}^d_+$ and $z = (z_1, \ldots, z_d) \in \RR^d$,  the monomial $z^k$ is defined as 
$$
z^k = z_1^{k_1}\cdots z_d^{k_d},
$$
and where the functions $(g_k)_k$ are differentiable.

\vskip .25cm

We are now ready to define the object of our study:

\begin{definition}\label{def:poly} A polynomial term structure model is the collection of functions $R,G, b, a$ satisfying Hypotheses (PDE), (SDE)
and (POLY) along with a family of weak solutions $(\Omega, \ff, \QQ; Z, W)$ indexed by the initial point $Z_0 = z \in I$.
A polynomial model is non-degenerate if the coefficients  $(g_k)_k$ are linearly independent functions.
\end{definition}

This work is inspired by the interest rate model of Siegel \cite{S}.  He
showed that for all integers $d \ge 1$ there exist   
an explicit affine functions $R$ and explicit quadratic functions $b$, 
such that Hypothesis (PDE) is satisfied by a 
function $G$ such that $G(\tau, \cdot)$  is affine for all $\tau \ge 0$.
Note that in this case $\partial_{z_i z_j} G$ vanishes identically, and 
hence the function $a=\sigma \sigma^\top$ need not be specified to
verify the partial differential equation.  Furthermore,
 it was shown that for a certain choice of $\sigma$ 
that the corresponding  the stochastic differential equation
 has a non-explosive solution valued in the bounded state-space 
$$
I = \left\{ (z_1, \ldots, z_d) : z_i > 0 \mbox{ for all $i$ and } \sum_i z_i < 1 \right\}.
$$

We mention  also the  Brody--Hughston rational  affine  model \cite{BH}.  Working 
under the objective measure $\PP$, the state price density is modelled
$V_t = \alpha(t) + \beta(t) M_t$ where $\alpha$ and $\beta$ are deterministic
functions and $M$ is a $\PP$-martingale.  Such rational affine models have been 
extended by Akahori--Hishida--Teichmann--Tsuchiya \cite{AHTT},  Filipovi\'c-Larsson--Trolle \cite{FLT}
and Macrina \cite{Macrina} among others.  We show in 
section~\ref{se:extensions} that the Brody--Hughston model fits in our 
time-inhomogeneous polynomial framework considered here.

Just as the Brody--Hughston model and the Siegel model described above, most of the polynomial
models  of this paper (but not all -- see section \ref{se:unbounded}) have the property that the spot interest rate is \textit{bounded}.
This stands in contrast to many familiar models, such as the Vasicek and Cox--Ingersoll--Ross models.
Nevertheless, the range of the spot interest rate can be expressed easily in terms of the
model parameters, and hence the range can be calibrated to any desired (finite) width.

Finally, a related work is that of Cuchiero, Keller-Ressel \& Teichmann \cite{CKT},
who study a class of time-homogeneous Markov process $Y$  with the property that the $n$-th (mixed)
 moments can be expressed as a polynomial of the initial point $Y_0$ of degree at most $n$.
Indeed, consider the $d=1$ case and let $F_n$ be the family of polynomials of degree at most $n$:
$$
F_n = \left\{P: P(z) = \sum_{k=0}^n p_k z^k, \ p_k \in \RR \right\}.
$$
They study the processes $Y$ that have the property that
for \textit{any} degree $n$ and \textit{any} polynomial $P \in F_n$,  for all $t \ge 0$ there exists
a polynomial $Q \in F_n$ such that
\begin{align*}
\mathbb{E}[P(Y_t)|Y_0 = y] = Q(y).
\end{align*}
In contrast, in this work we study  processes $Z$ that have the property that
for a \textit{fixed} degree $n$ and a \textit{fixed} function $R$,
for all $t \ge 0$ there exists a polynomial $P =G(t,\cdot) \in F_n$ such that
$$
\EE[ e^{-\int_0^t R(Z_s) ds} | Z_0 = z] = P(z).
$$
In particular, their results
do not imply ours, or vice versa. For further existence results for multi-dimensional
polynomial preserving processes, 
consult the recent paper of  Filipovi\'c and Larsson \cite{FL}.

In the remainder of this article is arranged as follows.  In section \ref{se:pde},
we show that the analytic hypothesis that the bond price function $G$ satisfies
a certain partial differential equation and the algebraic hypothesis
that  $G$ can be expressed as a polynomial of the factor forces
the interest rate function $R$ and the coefficients of the factor dynamics $b, a$ to be low-degree polynomials
of the factor.  Furthermore, we focus on dimension $d=1$ to explicitly spell out the 
 linear constraints these two hypotheses place on the coefficients of these polynomials.  
In section \ref{se:sde}   we provide  a complete classification scalar polynomial
 models satisfying the 
 probabilistic assumption that the corresponding stochastic differential equation has
a non-explosive solution valued in a bounded interval. 
In section~\ref{se:spectral} we present a spectral 
representation of the bond prices in the context of scalar polynomial models.  
In section \ref{se:examples} we 
consider a  concrete  example of a parametrised family of 
polynomial models which generalise in a certain sense the exponential affine models.
Finally in section \ref{se:extensions}, we briefly discuss   a Hull--White-type extension where the coefficients are allowed to be time dependent.
The appendix  contains an easy-to-check formulation
of Feller's test of explosion for stochastic differential equations for stochastic
differential equations with analytic coefficients, which might have independent interest.

\section{An algebraic result} \label{se:pde} 
This section contains one of the
 main result of this paper, a classification of 
 models that satisfy the analytic Hypothesis (PDE) that the pricing function $G$ solves
a particular partial differential equation, in addition to having the extra
structural property of Hypothesis (POLY) that $G(\tau, \cdot)$ is a polynomial of fixed degree.
To more clearly see the structure of the argument we consider only the 
time-homogeneous case  in this section. The time-inhomogeneous  case is considered in section \ref{se:extensions}.
The following theorem is of a purely algebraic nature. Indeed,  we are 
waiting until the following Section \ref{se:sde} to enforce the probabilistic Hypothesis (SDE).

\begin{theorem}\label{th:main1}  Suppose the functions $R, G, b, a$ satisfy Hypotheses (PDE) and (POLY)
where the degree of $G(\tau, \cdot)$ is at most $n \ge 1$. Furthermore, suppose the coefficient functions $(g_k)_k$
are linearly independent.

\noindent Case $n=1$.   The function $R$ is a polynomial of degree at most one, for each $i$ the function $b_i$ is a polynomial
of degree at most two, and the function $a$ is unrestricted.

\noindent Case $n\ge 2$.  The function $R$ is a polynomial of degree at most two, for each $i$ the function
$b_i$ is a polynomial of degree at most three, and for each $i,j$ the function 
$a_{i,j}$ is a polynomial of degree at most four. 
\end{theorem}

\begin{remark}
In light of  Filipovi\'c's maximal degree theorem for exponential polynomial models, it might come as a
surprise that the degree $n$ is not constrained for polynomial models.  
\end{remark}

\begin{proof} Fix $n \ge 1$, and define the following set of indices 
$$
K  = \{ k \in \mathbb{Z}^d_+ : k_1 + \ldots +k_d \le n \}.
$$ 
Hypothesis (PDE) gives rise to the condition 
\begin{equation}\label{eq:md}
\sum_{k \in K } \dot{g}_k(\tau ) z^k = \sum_{k \in K_n} g_k(\tau) A_k(z) \mbox{ for all } (t,z) \in \RR_+ \times I
\end{equation}
 where, for $k \in K $, the functions $A_k$ are defined as
$$
A_k(z) = \sum_{1 \le i \le d} b_i(z) \partial_{z_i} z^k  + \frac{1}{2} \sum_{1 \le i,j \le d} a_{ij}(z)
\partial_{z_i z_j}  z^k  - R(z)z^k.
$$

As in the introduction, for $m \ge 0$ define the notation
$$
F_m = \left\{ P:I \to \RR, P(z) = \sum_{k \in K_m} p_k z^k, \quad p_k \in \RR \right\}
$$
to be the family of polynomials in $d$ variables of total degree less or equal to $m$.
Since $I$ is open but not empty, the values of the function $P\in F_m$ uniquely determines its coefficients $(p_k)_k$.

First we show that the functions $A_k \in F_n$ are polynomials for all $k \in K_n$. Let $N = { n+d \choose n}$ 
be the cardinality of index set $K$. Since the functions $(g_k)_k$ are linearly independent, 
we can find $N$ distinct times $\tau_1, \ldots, \tau_N$ independent of $z$ such that the matrix with $i$-th column formed by vector $(g_k(\tau_i) , k \in K)$ is non-singular. Now fix any $z$, we can rewrite   condition \eqref{eq:md} as a set of $N$ simultaneous linear equations with $N$ unknowns $A_k(z)$. Therefore the solution exists and is unique and can be written as linear combinations of the monomials $z^k$.
In particular,  all of the $A_k(z)$ are polynomials in $d$ variables of total degree less or equal to $n$.

In what follows, let $\{ e_1, \ldots, e_d\}$ be the standard basis of $\RR^d$, so all the $i$th component of the
vector $e_i$ is one and the other
components are zero.

\noindent \textit{Case} $n=1$.  Since we must have $A_k(z) \in F_1$ for all $k \in K =\{0, e_1, \ldots, e_d\}$, we can conclude for any $1\leq i  \leq d$
\begin{align*}
A_0(z) &= -R(z) && \in F_1 \\
A_{e_i}(z) &= b_i(z) - z_i R(z) && \in F_1 
\end{align*}
This implies $R$ is affine, and hence $b_i(z) = A_{e_i}(z) + z_i R(z)$ is quadratic for all $i$.

\noindent \textit{Case} $n \ge 2$.  
Since we must have $A_k(z) \in F_n$ for all $k \in K$, we can conclude for any $1\leq i,j \leq d$
\begin{align*}
A_0(z) &= -R(z)  & \in F_n \\
A_{e_i}(z) &= b_i(z) - z_i R(z)  & \in F_n \\
A_{e_i+e_j}(z) &= b_i(z) z_j + b_j(z) z_i + a_{ij}(z) - z_i z_j R(z)  & \in F_n
\end{align*}
Therefore we may conclude that $R \in F_n$ that $b_i = A_{e_i} + z_i R \in F_{n+1}$ and 
$a_{ij} = A_{e_i+e_j} + z_i z_j R -  b_i  z_j + b_j z_i  \in F_{n+2}$. 
In particular,  the functions $R(z), b_i, a_{ij}$ are polynomials. On the other hand since
\begin{align*}
A_{n e_i}(z) = n z_i^{n-1}b_i(z) + \frac{n(n-1)}{2} z_i^{n-2}a_{ii}(z) - z_i^n R(z) \in F_n
\end{align*}
by cancelling the $z_i^{n-2}$ factor, we may deduce that 
\begin{equation} \label{eq:md1}
nz_ib_i(z) + \frac{n(n-1)}{2} a_{ii}(z) - z_i^2 R(z) \in F_2
\end{equation}
Similarly by considering $A_{(n-1)e_i}$ and $A_{(n-2)e_i}$, we get
\begin{equation} \label{eq:md2}
(n-1)z_ib_i(z) + \frac{(n-2)(n-1)}{2} a_{ii}(z) - z_i^2 R(z) \in F_3
\end{equation}
\begin{equation} \label{eq:md3}
(n-2)z_ib_i(z) + \frac{(n-2)(n-3)}{2} a_{ii}(z) - z_i^2 R(z) \in F_4
\end{equation}
Subtracting equation \eqref{eq:md1} from equation  \eqref{eq:md2} and subtracting equation  \eqref{eq:md2} from equation \eqref{eq:md3} yields
\begin{align*}
z_ib_i(z) + (n-1) a_{ii}(z)  \in F_3 \\
z_ib_i(z) + (n-2) a_{ii}(z)  \in F_4 \\
\end{align*}
Subracting once more yields $a_{ii} \in F_4$, and hence $ b_i \in F_3$. 
Substituting this into equation \eqref{eq:md1} yields $R \in F_2$.  

Finally, considering
$A_{(n-1)e_i + e_j}$ as above yields 
$$
(n-1)z_i z_j b_i(z) + z_i^2 b_j(z) +  \frac{(n-2)(n-1)}{2} z_j a_{ii}(z) + (n-1)z_i a_{ij}(z) - z_i^2 z_j R(z) \in F_3
$$
from which the conclusion $a_{ij} \in F_4$ follows.
\end{proof} 

We now restrict attention to the scalar case to describe explicitly the constraints on the coefficients 
of the various polynomials appearing in Theorem \ref{th:main1}:

\begin{theorem}\label{th:main}  Suppose the dimension is $d=1$ and the function $G$ satisfies Hypothesis (POLY),
where the degree of $G(\tau, \cdot)$ is at most $n \ge 1$.  
Furthermore, assume  $R(z) = R_0 + R_1 z + R_2 z^2$, $b(z) = b_0+ b_1 z + b_2 z^2 + b_3 z^3$ and
$a(z) = a_0+ a_1 z + a_2 z^2 + a_3 z^3 + a_4 z^4$.

Then the functions $R, G, b, a$ satisfies Hypothesis (PDE) if 
\begin{equation}\tag{COEF}
 R_2 =  \tfrac{n}{2} b_3 = - \tfrac{n(n-1)}{2} a_4 \mbox{ and } R_1 = n b_2 + \tfrac{n(n-1)}{2} a_3.
\end{equation}
and $(g_0, \ldots, g_n)$ solves the system of linear ordinary differential equations
\begin{align}
\dot{g_k} = &  \left( (k-2) b_3    + \frac{(k-2)(k-3)}{2} a_4 -R_2 \right)g_{k-2} \nonumber \\
&  + \left( (k-1) b_2 + \frac{(k-1)(k-2)}{2} a_3 - R_1 \right)g_{k-1}  
 + \left( kb_1 + \frac{k(k-1)}{2} a_2 - R_0 \right)g_k  \tag{ODE} \\ 
& + \left( (k+1)b_0 + \frac{k(k+1)}{2} a_1 \right) g_{k+1} +  \frac{(k+2)(k+1)}{2} a_0 \ g_{k+2} , \mbox{ for } 0 \le k \le n, \nonumber \\
g_k(0) & = \left\{ \begin{array}{ll} 1 & \mbox{ if } k= 0 \\ 0 & \mbox{ if } k \ge 1 \end{array} \right. \nonumber
\end{align}
  where we interpret $g_{-2} = g_{-1} = g_{n+1} = g_{n+2} = 0$.
 
Conversely if the functions $R, G, b, a$ satisfies Hypothesis (PDE) and the functions $(g_k)_k$ are linearly
independent, then the coefficients of the polynomials $R, b, a$ satisfy equation (COEF) and the coefficients
$(g_k)_k$ of the polynomial $G$ satisfy equation (ODE).
\end{theorem} 

To better understand the statement of Theorem \ref{th:main}, we introduce some notation
that we will use in the proof as well as in the sequel.  Fix $n \ge 1$, and let 
$\mathbf{L} = (L_{i,j})_{i,j = 0}^n$ be the $(n+1) \times (n+1)$ matrix with entries
$$
L_{j+k,j} = j b_{k+1} + \tfrac{j(j-1)}{2}  a_{k+2} - R_{k}
$$
and where $R_k= b_k = a_k= 0$ when $k < 0$ and $R_k = b_{k+1} = a_{k+2} = 0$ when $k > 2$.  For instance,
when $n \ge 4$, the matrix has the form
$$
\mathbf{L} = \left( \begin{array}{clllll}
-R_0 & \quad b_0          & \quad \hskip .85cm  a_0    &                        &           				  &   \\
-R_1 & \quad b_1- R_0     & \quad 2b_0 + a_1          & \quad \hskip .85cm  3 a_0     &           				  &   \\
-R_2 & \quad b_2 -R_1     &  \quad 2b_1+ a_2 -R_0     & \quad 3b_0 + 3a_1            & \quad \hskip .85cm  6 a_0 &   \\
     & \quad b_3 -R_2     & \quad 2b_2 + a_3 - R_1    & \quad 3b_1+3a_2 -R_0         & \quad 4b_0 + 6 a_1 			& \ddots \\
     &              & \quad 2 b_3 + a_4  -R_2   & \quad 3b_2 + 3 a_3 -R_1      & \quad 4 b_1 + 6 a_2 -R_0 & \ddots \\
     &              &                     & \quad \hskip .9cm \ddots     & \quad \hskip .9cm \ddots & \ddots
		\end{array} \right).
$$
If  we defined the $\RR^{n+1}$ valued function $g = (g_0, \ldots, g_n)^\top$ then  equation (ODE)
becomes 
$$
\dot g = \mathbf{L} g, \ \ g(0) = (1, 0, \ldots, 0)^\top.
$$

For future reference, let 
$\mathbf{I}$ is the $(n+1)\times (n+1)$ identity matrix and let $\mathbf{Z}$ be the  $(n+1)\times (n+1)$
matrix defined by  
$$
 \mathbf{Z} = \left( \begin{array}{ccccc} 0 & 0 & 0 &  \\ 1 & 0 & 0 & \ddots \\ 0 & 1 & 0 & \ddots \\   & \ddots & \ddots & \ddots \end{array} \right),
$$
so that $Z_{ij} = \delta_{i,j+1}$, where $\delta$ is the Kronecker delta.
Note that if we define the operation $\hat{}: \RR^{n+1} \to F_n$ by the formula
$$
\hat p(z) =  \sum_{k=0}^n p_k z^k = (1, z, \ldots, z^n) p
$$ 
for a column vector $p = (p_0, \ldots, p_n)^\top$, then
$$
z \hat p(z) = \widehat{\mathbf{Z} p}(z) + p_n z^{n+1}.
$$
Similarly, let
$$
 \mathbf{D} =  \left( \begin{array}{ccccc} 0 & 1 & 0 &  \\ 0 & 0 & 2 & \ddots \\ 0 & 0 & 0 & \ddots \\   & \ddots & \ddots & \ddots \end{array} \right)
$$
so that  $D_{ij} = i \ \delta_{i,j-1}$  and in particular
$$
\hat p'(z) =  \widehat{\mathbf{D} p}(z).
$$
With this notation, we have the formula
$$
\mathbf{L}   = b(\mathbf{Z}) \mathbf D + \tfrac{1}{2} a(\mathbf{Z}) \mathbf{D}^2 - R(\mathbf{Z}) .
$$
Letting $\mathcal L$ be the differential operator such that
$$
\mathcal L P(z) = b(z) P'(z) + \frac{1}{2} a(z) P''(z) - R(z) P(z)
$$
for a polynomial $P$, we have
\begin{align*}
\mathcal L \hat p(z)  = & \widehat{ \mathbf{L}p} (z) +  (n b_2 + \tfrac{n(n-1)}{2} a_3 - R_1)p_n z^{n+1} \\
& + ((n-1) b_3 + \tfrac{(n-1)(n-2)}{2} a_4 - R_2)p_{n-1} z^{n+1} \\
& + (n b_3 + \tfrac{n(n-1)}{2}a_4 - R_2) p_n z^{n+2}
\end{align*}

\begin{proof}  
Let $G$ satisfy Hypothesis (POLY) so that
$$
G(\tau, z) = \widehat{ g(\tau)}(z)
$$ 
in the notation introduced above. 
Notice that
\begin{align*}
\mathcal L G - \partial_\tau G =& \widehat{ \mathbf{L} g - \dot{g}}  + (n b_3 + \tfrac{1}{2}n(n-1) a_4 - R_2) g_n z^{n+2}  \\
&+  (n b_2 + \tfrac{1}{2}n(n-1) a_3 - R_1) g_n z^{n+1}  
  +  ((n-1) b_3 + \tfrac{1}{2}(n-1)(n-2) a_4 - R_2) g_{n-1} z^{n+1}
\end{align*} 
Clearly, if the equations (COEF) and (ODE) hold, then $G$ satisfies Hypothesis (PDE). 

Conversely, suppose  $\partial_\tau G = \mathcal L G$. The left hand side of the above equation vanishes
and the term   $\widehat{ \mathbf{L} g - \dot{g}}$ is of degree at most $n$.  Hence, 
the coefficient of $z^{n+2}$  must vanish yielding
\begin{equation}\label{eq:first}
  n b_3 + \frac{n(n-1)}{2} a_4 =  R_2
\end{equation}
Similarly,  the linear independence of $g_{n-1}$ and $g_n$ implies that the coefficients of each of the two $z^{n+1}$ terms vanish,
yeilding $$
 n b_2  + \frac{n(n-1)}{2} a_3 =  R_1.
$$
and
\begin{equation}\label{eq:second}
(n-1) b_3 + \frac{(n-1)(n-2)}{2} a_4 =  R_2
\end{equation}
Note that equations \eqref{eq:first} and \eqref{eq:second} together are equivalent to equation (COEF).
Finally, we are left with $ \widehat{ \mathbf{L} g} = \hat{\dot{g}} $ identically, implying $\mathbf{L}  g = \dot g$ as
as claimed.
\end{proof}

\section{Some probabilistic results} \label{se:sde}
In this section we include some results related to the probabilistic
assumption that a certain stochastic differential equation has a non-explosive
solution.

\subsection{Bounded state space}\label{se:bounded}
We now argue that there are good reasons to make the
further assumption that the state space $I$ of the factor process $Z$ in a
 polynomial model  is   bounded, at least in the one-dimensional case.

Recall that we aim to model the price $P_{t,T}$ at time
$t$ of a zero-coupon bond of maturity $T$ by the formula $P_{t,T} =G(T-t,Z_t)$ where
$Z_t$ is the economic factor at time $t$.  Since the payout of the bond is its face 
value  $P_{T,T} = 1$, it is reasonable to assume
that the bond prices are bounded. Indeed, to avoid
a buy-and-hold arbitrage, one must have $P_{t,T} > 0$ for all $0 \le t \le T$;
furthermore, assuming the existence of a bank account continuously paying
the spot interest rate $r_t$ and assuming that this interest
rate is bounded from below in the sense that $r_t \ge - C$ for all $t \ge 0$ for some
constant $C > 0$, then there would be a buy-and-hold  arbitrage unless  $P_{t,T} \le e^{ (T-t)C}$
for all $0 \le t \le T$.

  The above discussion motivates considering the  additional hypothesis that the bond prices are 
bounded:
\begin{proposition}\label{th:bounded} Consider a non-degenerate polynomial model in dimension $d=1$.    Then the following statements are equivalent:
\begin{enumerate}
\item The function  $G(\tau, \cdot)$ is bounded on $I$ for all $\tau \ge 0$.
\item The function $R$ is bounded  on $I$.
\item The interval $I \subset \RR$ is bounded.   
\end{enumerate}
 If any one (and therefore all) of the statements holds, then 
$$
G(\tau, z) =   \EE[ e^{-\int_0^\tau R(Z_s) ds } | Z_0 = z] \mbox{ for all } (t, z) \in \RR_+ \times I.
$$
\end{proposition}

\begin{proof} (1) implies (3). By Hypothesis (POLY)  the 
function  $G(\tau, \cdot)$ is a polynomial for all $\tau \ge 0$.  Furthermore, if it were the case that $G(\tau, \cdot)$
was a constant for all $\tau \ge 0$, then we would have $g_k(\tau) = 0$ for all $k \ge 1$ and $\tau \ge 0$, 
contradicting the assumption   that the coefficients are linearly independent.  Hence,
there exists a $\tau > 0$ such that $G(\tau, \cdot)$ is non-constant.  
We are done since non-constant polynomials in 
one real variable are unbounded on unbounded intervals.

(2) implies (3).  By Theorem \ref{th:main} the function $R$ is a polynomial.  If it were the case that
$R$ was constant, then one (and therefore the only) solution to the system of linear equations (B) would be $g_0(\tau) = e^{-R_0 \tau}$
and $g_k(\tau) = 0$ for all $k \ge 1$ and $\tau \ge 0$.  Again this would contradict the assumption of linear independence
of the functions $(g_k)_k$.
Hence $R$ is a non-constant polynomial and we are done.

(3) implies both (1) and (2). This is obvious, since the functions $G(\tau, \cdot)$ and $R$ are polynomials.

Fix $z \in I$, and let $Z$ solve the stochastic differential equation 
with $Z_0=z$.  Also fix a time horizon $\tau > 0$.
As mentioned in the introduction, 
since $G$ satisfies the partial differential equation, then by It\^o's formula we know that the process
$M = (M_t)_{0 \le t \le \tau}$ defined by
$$
M_t = e^{-\int_0^t R(Z_s) ds} G(\tau-t, Z_t)
$$
is a local martingale.  By assumptions (1) and (2), the
 process $M$ is bounded by a constant.

In particular, the bounded local martingale $M$ is a true martingale by the dominated convergence theorem,
and hence
\begin{align*}
G(\tau,z)  = M_0 & = \EE[ M_\tau ] \\
& =  \EE[ e^{-\int_0^\tau R(Z_s) ds}  ]
\end{align*}
as desired.
\end{proof} 
  
\begin{remark}
Note that in higher dimensions, non-constant polynomials may be bounded on unbounded sets.
For instance, consider the polynomial 
$$
P(z_1,z_2) = z_1 - z_2
$$ on the unbounded set
$$
I = \{ (z_1,z_2):  |z_1 - z_2| \le 1 \}.
$$
\end{remark}

\begin{remark} Notice that the boundedness of the bond pricing function $G$   does \textit{not}
 imply the boundedness of the state space $I$ in the case of \textit{exponential} polynomial models such as Cox--Ingersoll--Ross.
\end{remark}
 
The following corollary of Propostion \ref{th:bounded} also serves somewhat as a converse:

\begin{corollary}\label{th:bounded2} Consider a non-degenerate polynomial model in dimension $d=1$, such that
the function $R$ is bounded from below on $I$ and $$
G(\tau, z) =   \EE[ e^{-\int_0^\tau R(Z_s) ds } | Z_0 = z] \mbox{ for all } t \ge 0, z \in I.
$$
Then the interval $I$ is bounded and the function $R$ is  bounded from above on $I$.
\end{corollary}

\begin{proof} From the formula, it is clear that the function $G$ is bounded from below by zero.
And since $R$ is bounded from below, the function $G(\tau, \cdot)$ is bounded from above. 
The conclusion follows from Proposition \ref{th:bounded}
\end{proof}

Before closing this section, we consider a consequence of Theorem \ref{th:main} in the context
of bounded scalar polynomial models:

\begin{proposition}  Consider a non-degenerate polynomial model in dimension $d=1$ where 
the interval $I$ is bounded and $G(\tau, \cdot)$ is of degree at most $n$.  
Let $P$ is a polynomial at most $n$ and let
$$
\EE[ e^{-\int_0^{\tau} R(Z_s) ds} P(Z_t) | Z_0=z ] = Q(\tau, z) \mbox{ for } \tau \ge 0, z \in I.
$$
Then for all $\tau \ge 0$, the function $Q(\tau, \cdot)$ is a polynomial of degree at most $n$.
\end{proposition}

\begin{proof} Suppose $P$ can be written as $P(z) = \sum_{k=0}^n p_k z^k$.   Let
$q$ be the unique $\RR^{n+1}$ valued solution of
$$
\dot q = \mathbf{L} q, \ \ q(0) = p.
$$
Let $Q(\tau, z) = \widehat{q(\tau)}(z)$ for $\tau \ge 0$.  Note that $Q$ solves the partial differential equation
\begin{align*}
&\partial _{\tau} Q = \mathcal L Q \mbox{ on } \RR_+ \times I \\
&Q(0, z) = P(z) \mbox{ for } z \in I.
\end{align*}
By the same argument of the proof of Proposition \ref{th:bounded} we conclude
$$
\EE[ e^{-\int_0^{\tau} R(Z_s) ds} P(Z_t) | Z_0=z ] = Q(\tau, z) \mbox{ for } \tau \ge 0, z \in I.
$$
as desired.
\end{proof}

\subsection{An unbounded example}\label{se:unbounded}
The message of Proposition \ref{th:bounded} is that
the assumptions that the bond prices are bounded 
and that the bond prices are polynomials in a scalar factor together imply 
that the spot interest rate is bounded. 

In this section, we show by example that there exists an example of a
polynomial model  where the spot rate, and therefore the bond prices,
are unbounded.  The point is not to suggest that this particular model
is a good model of real world interest rates, but rather to show that 
the hypotheses that define polynomial models do not by themselves imply
boundedness. That is to say, if one wants to a model to imply bounded bond prices, then it 
is necessary to add boundedness as an additional assumption.

Let the state space be $I = (0, \infty)$, and the degree be $n=2$,  the coefficient functions
be given by
$$
a(z) = z^2, \ b(z) = - \frac{1}{2} z^2, \ R(z) = - z,
$$
and the bond pricing function be
\begin{align*}
G(\tau,z)  = 1 + \tau z + \frac{1}{2}(e^\tau - \tau -1 ) z^2.
\end{align*}
Note that 
$$
\partial_\tau G = b \partial_z H  + \frac{1}{2}a  \partial_{zz} G - R G
$$
with initial condition 
$$
G(0,z) = 1
$$
so we are in the setting of Theorem \ref{th:main}.   Finally note
that the unique strong solution of the stochastic differential equation 
$$
dZ_t = - \frac{1}{2} Z_t^2 \ dt + Z_t \ dW_t, \ Z_0 = z
$$
is given by the formula 
$$
Z_t = \frac{   z e^{W_t - t/2} }{1 + \frac{z}{2} \int_0^t e^{W_s - s/2}ds }.
$$
Therefore, these data constitute a polynomial model according to Definition \ref{def:poly}
where the spot rate is unbounded.

Finally, notice the identity 
\begin{align*} 
\EE( e^{\int_0^\tau Z_s ds} ) & = \EE \left[ \left( 1 + \frac{z}{2} \int_0^\tau e^{W_s - s/2}ds \right)^2 \right] \\
& = G(\tau,z)
\end{align*}
which can be verified by explicit calculation. 

\begin{remark} This example does not violate Corollary \ref{th:bounded2} since in this case,
the function $R$ is bounded \textit{from above} but not from below.
\end{remark}

\begin{remark}
We mention here an interesting (though a bit tangential) observation regarding the above
example.  It is easy to see that that the process $Z$ introduced is such that
the process $Y = e^{Z}$ defines a local martingale with dynamics
$$
dY_t = Y_t \log(Y_t) \ dW_t, \ \ Y_0 = e^z.
$$ 
 It is slightly less
obvious that the process $Y$ is a strictly local martingale.
  See, for instance, the paper of   Goodman \cite{goodman} for related results.
\end{remark}

\subsection{A form of Feller's test}
As argued in Section \ref{se:bounded}, 
there are economic reasons to   consider polynomial models in which
the factor process takes values in a bounded interval.
Therefore, in this subsection we consider
solutions to the scalar stochastic differential equation
$$
dZ_t = b(Z_t) dt  + \sigma(Z_t) dW_t,
$$
which live in a bounded state space $I = (z_{\mathrm{min}}, z_{\mathrm{max}})$.
Furthermore, in light of Theorem \ref{th:main},   we assume that the coefficients 
 $b$ and $\sigma^2$ are polynomials.

To avoid trivial complications, we assume
$$
\sigma(z_{\mathrm{min}}) = 0 =\sigma( z_{\mathrm{max}}) \mbox{ and } \sigma(z) > 0 \mbox{ for }  z_{\mathrm{min}} 
< z < z_{\mathrm{max}}.
$$
Note that the coefficient $b$ is Lipschitz on the closed
interval $[z_{\mathrm{min}}, z_{\mathrm{max}}]$, while the 
coefficient $\sigma$ is Lipschitz on any interval 
$[z_{\mathrm{min}}+1/N, z_{\mathrm{max}} - 1/N]$
for $N > 1$ large enough.
Therefore, for every $z \in   (z_{\mathrm{min}}, z_{\mathrm{max}})$   the stochastic differential equation has a unique
nested family of 
strong solutions $(Z_{t,N})_{t \in [0,T_N]}$ with $Z_{0,N} = z$, where
$$
T_N = \inf \{ t>0: Z_t \notin (z_{\mathrm{min}}+1/N, z_{\mathrm{max}} - 1/N) \}.
$$
The explosion time $T$ is then defined as
$$
T = \sup_N T_N.
$$
We are interested in the case where the solution is non-explosive
in the sense that $T = \infty$ almost surely.

 The classical necessary and sufficient
conditions on the functions $b$ and $\sigma$  is Feller's test of explosion.
Motivated by Theorem \ref{th:main} we adapt Feller's test to the case
where the functions $b$ and $\sigma^2$ are polynomials.  It is likely that 
the following result is well-known, but we were unable to locate a reference
in the literature.

\begin{theorem}\label{th:feller}
Let $b$ and $a$ be real analytic. Furthermore, assume
$$
a(z_{\mathrm{min}}) = 0 =a(z_{\mathrm{max}}) \mbox{ and } a(z) > 0  \mbox{ for }  z_{\mathrm{min}} 
< z < z_{\mathrm{max}}.
$$
Letting $\sigma = \sqrt{a}$, 
there exists a unique non-explosive strong solution $Z$ of the stochastic differential equation
$$
dZ_t = b(Z_t) dt  + \sigma(Z_t) dW_t 
$$
taking values in the interval $(z_{\mathrm{min}}, z_{\mathrm{max}})$  if and only if
$$
2b(z_{\mathrm{min}}) - a'(z_{\mathrm{min}}) \ge 0 \ge 2b(z_{\mathrm{max}}) - a'(z_{\mathrm{max}}).
$$
\end{theorem}

\subsection{A canonical parametrisation of scalar polynomial models}\label{se:canon}
We are now in a position 
to characterise the  range of admissible parameters
for which there exists a bounded  scalar polynomial model.

In light of Theorem \ref{th:main1}, we may assume the degree of
 $b$ is at most three and the degree of $\sigma^2$ is at most four.

By applying an affine transformation to the state variable,   there is no loss
of generality in fixing the state space $I$ to be any finite
interval.  Therefore, to simplify some calculations,
 in this section we will set $I = (-1,1)$ and will refer to this
as the canonical state space in the sequel.  

In order to enforce the condition $\sigma(-1) = 0 = \sigma(1)$
we rewrite $\sigma^2$ as a product of $(1-z^2)$ and polynomial
of degree of at most two.

\begin{proposition}\label{th:no-exp}
Let
\begin{align*}
b(z) &= b_0 + b_1 z + b_2 z^2 + b_3 z^3 \\
\sigma^2(z) &=  (1-z^2) (c_0 + c_1 z + c_2 z^2).
\end{align*}
The  stochastic differential equation
$$
dZ_t = b(Z_t) dt  + \sigma(Z_t) dW_t \ \ Z_0=z
$$
has a non-explosive solution valued in the open interval $(-1,1)$
for every initial condition $-1 < z < 1$ if and only if all of the following conditions hold
\begin{itemize}
\item 
$b_1 + b_3 + c_0 + c_2  \le - | b_0 + b_2 + c_1 |$;
\item  $c_0 > 0$; and
\item either  $   |c_1|- c_0 \le c_2 \le c_0  $ or  $    c_2 > \max\{ c_0,  \tfrac{1}{4} c_1^2 \} $
\end{itemize}
\end{proposition}

We prove this result via two lemmas.

\begin{lemma}
Let $a$ and $b$ be as in Proposition \ref{th:no-exp}.  
 We have 
$$
2b(-1) - a'(-1) \ge 0 \ge 2b(1) - a'(1)
$$
 if and only if
$$
b_1 + b_3 + c_0 + c_2  \le - | b_0 + b_2 + c_1 |.
$$
\end{lemma}

\begin{proof}  We have
$$
2b(z) - a'(z) = 2 [ b_0  + (b_2 + c_1) z^2  ]  + 2 z [ b_1 + c_0 + (b_3+c_2) z^2] - (1-z^2)( c_1 + 2 c_2 z)
$$
from which the conclusion quickly follows.
\end{proof} 

\begin{lemma}  
 We have 
$$ 
c_0 + c_1 z + c_2 z^2 > 0 \mbox{ for all } -1 < z < 1 
$$
 if and only if
\begin{itemize}
\item  $c_0 > 0$; and
\item either  $   |c_1|- c_0 \le c_2 \le c_0  $ or  $    c_2 > \max\{ c_0,  \tfrac{1}{4} c_1^2 \} $
\end{itemize}
\end{lemma}

\begin{proof}
For a fixed triplet $(c_0, c_1, c_2)$ let $c(z) =  c_0 + c_1 z + c_2 z^2$.

To prove necessity, we first suppose that  $c(z) > 0$ for all $-1 < z < 1$.
Note that $c(0) = c_0$ implying  that $c_0 > 0$. Furthermore,
 by continuity, we have $c(\pm 1) \ge 0$   implying that $ c_2 \ge |c_1| - c_0$. 

Now  consider the case where $c_2 >  c_0$. 
Letting
$$
z_0 = \sqrt{\frac{c_0}{ c_2} }  
$$
we have $0 < z_0 < 1$ and 
$$
c(z_0) =  (c_1 +  2 \sqrt{  c_0 c_2})z_0.
$$
Hence the condition $c_1 > - 2 \sqrt{  c_0 c_2}$ is necessary.  By considering
$c(-z_0)$ we see that  the condition  $c_1 <  2 \sqrt{  c_0 c_2}$ is also necessary. 

Now to prove sufficiency, first suppose $c_0 > 0$  and $   |c_1|- c_0 \le c_2 \le c_0$
 Note 
that 
\begin{align*}
c(z) &\ge c_0 - |c_1| |z| + c_2 z^2 \\
& \ge  c_0 - (c_0 + c_2) |z| + c_2 z^2 \\
& = (c_0 - |z|c_2 )(1- |z|)
\end{align*}
and hence $c(z) > 0$ whenever $|z| < 1$. 

Finally, suppose $c_2 > c_0 > 0$ and $|c_1| <  2 \sqrt{  c_0 c_2}.$
Writing
\begin{align*}
c(z) &=  c_0 - \frac{c_1^2}{4 c_2} + c_2 \left( z - \frac{c_1}{2 c_2} \right)^2.
\end{align*}
we have $c(z) > 0$ for all $z$. 
\end{proof}

\begin{proof}[Proof of Proposition \ref{th:no-exp}] The claim follows from the two lemmas and
the version of Feller's test of Theorem \ref{th:feller}.
\end{proof}

\section{A spectral representation}\label{se:spectral}
We are in the setting of the scalar non-degenerate polynomial model
with the factor process taking values in a bounded open interval $I$.   
Recalling the notation $\mathbf{L}$ from Section \ref{se:pde}, we
note that the coefficient functions $g = (g_0, \ldots, g_n)^{\top}$ are the 
solution of the system of differential equations 
$$
\dot g = \mathbf{L} g, \  \ g(0) = (1, 0, \ldots, 0)^\top
$$
or equivalently $g(\tau) = e^{\mathbf{L} \tau} g(0)$.

It turns out that the matrix $\mathbf{L}$ has a nice property:
\begin{proposition}\label{th:realeigs} Let
$$
\mathbf{L}   = b(\mathbf{Z}) \mathbf D + \tfrac{1}{2} a(\mathbf{Z}) \mathbf{D}^2 - R(\mathbf{Z}) .
$$
where $R(z) = R_0 + R_1z + R_2z^2$, $b(z) = b_0 + b_1z + b_2 z + b_3z^2$ and $a(z) = (1-z^2)(c_0 + c_1 z + c_2 z^2)$ and the
coefficients satisfy
\begin{itemize}
\item $ R_2 =  \tfrac{n}{2} b_3 = - \tfrac{n(n-1)}{2} a_4 \mbox{ and } R_1 = n b_2 + \tfrac{n(n-1)}{2} a_3;$
\item 
$ |b_0 + b_2 + c_1 | \le - (b_1 + b_3 + c_0 + c_2)$;
\item  $c_0 > 0$; and
\item either  $   |c_1|- c_0 \le c_2 \le c_0  $ or  $    c_2 > \max\{ c_0,  \tfrac{1}{4} c_1^2 \} $
\end{itemize}
The eigenvalues $\lambda_0, \ldots, \lambda_n$ of $\mathbf{L}$ are real and satisfy 
$$
\lambda_i \le -\inf_{z \in I} R(z).
$$
for all $i$.
\end{proposition}

To prove Proposition \ref{th:realeigs}, we first prove a result
on the existence of an invariant measure which may have
independent interest.

\begin{proposition}\label{th:invmeas}  Suppose
the functions $b,a$ are polynomials, such that 
\begin{itemize}
\item $b(-1) > 0 > b(1)$ and 
\item $a(-1) = 0 = a(1)$ and
\item $a'(-1) > 0 > a'(1)$ and $a(z) > 0$ 
\item $a(z) > 0$ for all $- 1< z < 1$.
\end{itemize}
  Then there
exists a positive, integrable function $f$ satisfying the
differential equation
$$
bf = \frac{1}{2}(a  f)'
$$ 
with boundary conditions
$$
\lim_{z \downarrow -1} a(z) f(z) = 0 = \lim_{z \uparrow  1}a(z) f(z),
$$
\end{proposition}

\begin{remark}
If the function 
$f$ is normalised so that $\int_{-1}^1 f(z) dz = 1$, then $f$
is the unique invariant density for the diffusion $Z$ with drift $b$ and
volatility $\sigma = \sqrt{a}$. That is, if the initial condition
$Z_0$ is distributed with density $f$, then $Z_t$ has the
same distribution for all $t \ge 0$.
\end{remark}

\begin{proof} 
Now any positive solution to the differential equation is of form
$$
f(z) = \frac{C}{a(z)}e^{ \int_0^z \frac{2 b(s)}{a(s)} ds}.
$$
for $|z| < 1$, where $C>0$ is a constant. 

As in the proof of Theorem \ref{th:feller} we focus on the left-hand 
end point $z=-1$.  We must show that such an $f$ is integrable
and $a(t-1) f(t-1) \to 0$ as $t \downarrow 0$.   Writing
\begin{align*}
b(t-1) &= \beta + O(t) \\
a(t-1) & = \alpha t + O(t^2)
\end{align*}
where $\beta, \alpha > 0$, a routine calculation shows
that
$$
f(t-1) = O( t^{ \frac{2 \beta}{\alpha} - 1})
$$
from which the conclusion follows.
\end{proof}

\begin{proof}[Proof of Proposition \ref{th:realeigs}]   
Since $L$ varies continuously with the model parameters,
there is no loss of generality to assume that the parameters satisfy
\begin{itemize}
\item 
$ |b_0 + b_2 + c_1 | < - (b_1 + b_3 + c_0 + c_2)$;
\item  $c_0 > 0$; and
\item either  $   |c_1|- c_0 < c_2 \le c_0  $ or  $    c_2 > \max\{ c_0,  \tfrac{1}{4} c_1^2 \} $
\end{itemize}
  By Proposition \ref{th:invmeas} there exists
an invariant density $f$.
 
Consider the inner product on $\RR^{n+1}$
defined by
\begin{align*}
\langle p, q \rangle  &= \sum_{i=0}^n \sum_{j=0}^n p_i q_j \int_{-1}^1 z^{i+j} f(z) dz \\
&= \int_{-1}^1 \hat p(z) \hat q(z) f(z) dz
\end{align*}
where as before $\hat{}$ is the linear operator
such that
$$
\hat p(z) = \sum_{k=0}^n p_k z^k.
$$
Recall that 
$$
\widehat{\mathbf{L}p} = \frac{1}{2} a \hat p'' + b \hat p' - R \hat p.
$$
By integration by parts we have
\begin{align*}
\langle p, \mathbf{L} q \rangle & =  - \int_{-1}^1   [ \tfrac{1}{2} a \hat p' \hat q'  + R \hat p \hat q] f  \ dz \\
 &= \langle\mathbf{L} p,  q \rangle,
\end{align*}
 where we have used the boundary condition
$$
\lim_{z \downarrow -1} a(z) f(z) = 0 = \lim_{z \uparrow  1}  a(z) f(z).
$$
In particular,  we see that $\mathbf{L}$ is symmetric with respect
to this inner product and hence all eigenvalues are real.
The inequality
\begin{align*}
\langle p, \mathbf{L} p   \rangle & \le   - \int_{-1}^1     R   \hat p^2  f  \ dz \\ 
& \le - \inf_{-1 < z < 1} R(z) \langle p, p \rangle
\end{align*}
implies the claimed upper bound on the spectrum.
\end{proof}

\begin{remark}
Of course, the eigenvalues of the matrix $\mathbf{L}$ are the zeros of the characteristic polynomial
which has  degree $n+1$.  Since there exists formulae for the roots of polynomials up to degree four,
it is possible, at least in principle, to express   explicitly  the bond pricing function $G$ in a scalar polynomial model
in terms of the model parameters when $n \le 3$.

 When $n \ge 4$, there is little hope for explicit formulae for the function  $G$ in
terms of the model parameters.  However, note that the matrix $\mathbf{L}$ is sparse, in the sense that there are
at most five non-zero matrix entries per row.  In particular, the product of the matrix exponential $e^{\mathbf{L}\tau}$
and the vector $(1, 0, \ldots, 0)^\top$ can be computed efficiently, and hence the lack of explicit formulae is not necessarily
a prohibitive disadvantage.
\end{remark}

The proof of Proposition \ref{th:realeigs} shows that when there exists
an invariant density $f$, then 
$$
\mathbf{L}^{\top} M = M \mathbf{L}
$$
where $M = (M_{ij})_{ij}$ is the $(n+1) \times (n+1)$ positive definite
matrix with entries
$$
M_{ij} = \int_I z^{i+j} f(z) dz.
$$ 
  Suppose that the $n+1$  real eigenvalues
of the matrix $\mathbf{L}$ are
$\lambda_0,  \ldots, \lambda_n$.  Then the matrix $\mathbf{L}$ has the spectral 
decomposition
$$
\mathbf{L} = \sum_{i=0}^n \lambda_i \ u_i \ v_i 
$$
where $u_i$ is the right-eigenvector and $v_i$ the left-eigenvector associated to 
the eigenvalue $\lambda_i$, scaled such that
$$
 v_i  u_j = \left\{ \begin{array}{ll} 1 & \mbox{ if } i=j \\ 0 & \mbox{ if } i \ne j.  \end{array} \right.
$$
For convenience, we choose the normalisation
$$ 
u_i^\top M u_i  = 1, 
$$
and note that  the left- and right-eigenvectors are related by
$$
v_i =   u_i^\top M.
$$
Now given the $i$th right-eigenvector $u_i$
we can form the polynomial $\hat u_i(z) = \sum_{k=0}^n u_{i,k} z^k$. Note that $\hat u_i$
is an eigenfunction of the differential operator $\mathcal L$,  
and that the $i$th left-eigenvector $v_i$ is related to $\hat u_i$ by the formula
$$
v_{i,k} =  \int_I z^k \hat u_i(z) f(z) dz 
$$ 

In particular,  the bond pricing function takes
the form
$$
G(\tau,z) = \sum_{i=0}^n Q_i(z) e^{\lambda_i \tau}
$$
where the function $Q_i$ is the (at most) $n$ degree polynomial
$$
Q_i(z) = \hat u_i(z)  \int_I   \hat u_i(s) f(s) ds 
$$
That is to say, the bond price can be seen to be a linear combination of the bond prices
arising from $n+1$ models with constant interest rates $r = -\lambda_i$, where the coefficients $Q_i$
of the combination depend on the factor process.  Note that by setting $\tau=0$ we have
$$
\sum_{i=0}^n Q_i(z) = 1
$$
so it is tempting to think of the numbers $(Q_i(z))_i$ as probabilities; however, in general
$ Q_i(z) < 0$ for some $i$ and $z \in I$, so such an interpretation is not
always valid.

In the general case, where the parameters are such that no invariant density
exists, the matrix $\mathbf{L}$ is not necessarily diagonalisable. 
In this case, 
the bond pricing formula must be modified to 
\begin{equation}\label{eq:weighted}
G(\tau,z) = \sum_{i=0}^n Q_i(\tau,z) e^{\lambda_i x}
\end{equation}
 where now the weight functions $Q_i$  are polynomials in both $x$ and $z$
and can be computed from the Jordan decomposition.   An example
where the matrix $\mathbf{L}$ is not diagonalisable is
discussed in Section \ref{se:unbounded} -- though
strictly speaking, the setting is slightly different there since the state space for that
example is unbounded.

One consequence of   formula  \eqref{eq:weighted} 
is that the long maturity interest rate can calculated as
$$
\lim_{\tau \to \infty} - \frac{1}{\tau} \log G(\tau,z) = - \max_i \lambda_i.
$$
for all $z \in I$,   unless the coefficient $Q_i(\tau,z)$ of the maximum eigenvalue  is identically zero.

\section{An example} \label{se:examples}
In this section we explore a concrete realisation of 
a polynomial model.  The purpose of this account is as a
 proof of concept  and is not intended as an endorsement of this
particular model over others. 
In the general polynomial framework, the function $R: I \to \RR$, mapping
 the factor process to the spot interest rate, is a quadratic function.
In the following example, we assume that $R$ is affine.  By an affine
change of variables, we can and will take the spot rate itself as the factor
process. Note that this choice of parametrisation differs from the canonical
choice introduced in Section \ref{se:canon}.

The following proposition requires no proof in light of Theorem \ref{th:main}, Proposition \ref{th:bounded}
Theorem \ref{th:feller}.

\begin{proposition}  \label{th:example}
Given  real constants $\alpha, \beta, \gamma$, positive constants $\delta, \varepsilon$,
and a positive integer $n$ such that
\begin{align*}
& 2( \alpha \delta + \beta \gamma)  \ge (\delta+ \gamma \varepsilon)^2 \\
& \beta \ge \alpha \varepsilon \\
& n(n-1) \delta \varepsilon^2 = 2.
\end{align*}
For every $\rho$ in the interval $I = (-\gamma/\delta, 1/\varepsilon)$ there exists
a unique (non-explosive) $I$-valued strong solution $(r_t)_{t \ge 0}$ to the stochastic
differential equation
$$
dr_t = (\alpha - \beta r_t) dt + \sqrt{\gamma + \delta r_t} (1- \varepsilon r_t) dW_t, \ \  r_0= \rho
$$
with the property that there exist  differentiable functions $g_0, \ldots, g_n: \RR_+ \to \RR$ such that
$$
\EE[ e^{-\int_0^{\tau} r_s ds} | r_0= \rho] = \sum_{k=0}^n g_k(\tau) \rho^k 
$$
for all $\rho \in I$ and $\tau \ge 0$. Furthermore, the function $g=(g_0, \ldots, g_n)^\top$
is the unique solution to the linear ordinary differential equation
$$
\dot{g} = \mathbf{L} g, \ \ g(0) = (1, 0, \ldots, 0)^\top  
$$
 where the $(n+1)\times (n+1)$ matrix $\mathbf{L}$
is given by
$$
\mathbf{L} = - \mathbf{Z} +   (\alpha \mathbf{Z} - \beta \mathbf{I}) \mathbf{D} + \frac{1}{2}
(\gamma \mathbf{I} + (\delta - 2 \gamma \varepsilon)\mathbf{Z} +
(\gamma \varepsilon^2 - 2 \varepsilon \delta ) \mathbf{Z}^2 + \delta \gamma \mathbf{Z}^3) \mathbf{D}^2
$$
\end{proposition}

The above proposition could be compared to the following proposition on exponential
affine term structure models:

\begin{proposition}  \label{th:affine}
Given  real constants $\alpha, \beta, \gamma$ and a non-negative constant $\delta$, such that
\begin{align*}
& 2( \alpha \delta + \beta \gamma)  \ge \delta^2.
\end{align*}
For every $\rho$ in the interval $I$ defined by
$$
I = \left\{ \begin{array}{ll} (-\gamma/\delta, + \infty) & \mbox{ if } \delta > 0 \\
(-\infty, + \infty) & \mbox{ if } \delta = 0 \end{array} \right.
$$
 there exists
a unique (non-explosive) $I$-valued strong solution $(r_t)_{t \ge 0}$ to the stochastic
differential equation
$$
dr_t = (\alpha - \beta r_t) dt + \sqrt{\gamma + \delta r_t}  dW_t, \ \  r_0= \rho
$$
with the property that there exist  differentiable functions $h_0, h_1: \RR_+ \to \RR$ such that
$$
\EE[ e^{-\int_0^{\tau} r_s ds} | r_0= \rho] = e^{ h_0(\tau) + h_1(\tau) \rho}  
$$
for all $\rho \in I$ and $\tau \ge 0$. Furthermore, the functions $h_0, h_1$ satisfy the coupled
system of Ricatti equations:
$$
\begin{array}{lr}
\dot h_1   = - 1 - \beta h_1 + \frac{1}{2} \delta h_1^2,  & \quad  h_1(0) = 0 \\
\dot h_0   =  \alpha h_1 + \frac{1}{2} \gamma h_1^2, & \quad  h_0(0) = 0
\end{array}
$$
\end{proposition}
 
\begin{remark}
Recall that the Vasicek interest rate model is recovered from the more general affine
model of the above proposition by setting $\delta = 0$.  Similarly, the Cox--Ingersoll--Ross model
corresponds to $\gamma = 0$.
\end{remark}

\begin{remark}
Comparing propositions \ref{th:example} and \ref{th:affine} we see that
the dynamics of a certain class of exponential affine processes popularly 
used in interest rate modelling can be recovered  from 
a certain class of polynomial models by formally setting $\varepsilon = 0$ and
$n=\infty$.
\end{remark}

See the recent thesis \cite{si} of Cheng for calibrated examples of scalar 
polynomial term structure models from the class of examples exhibited in Proposition \ref{th:example}.

\section{Hull--White-type extension} \label{se:extensions}
In this section we consider a Hull--White type extension of the polynomial modelling framework.
As usual, by incorporating time-dependent parameters, we can hope to have a better model calibration. We introduce time dependency both in the dynamics of the factor process $(Z_t)_{t \geq 0}$ and the coefficient functions $(g_k)_k$. 
We first  establish an algebraic result similar to Theorems \ref{th:main} and \ref{th:main1} in this case.
We will then show
that the Brody--Hughston rational   model can be seen as an instance of this framework when the
degree is $n=1$.

\begin{theorem}\label{th:HW}
Let $\Delta = \{ (t,T): 0 \le t \le T \}$ and $I \subseteq \RR^d$ be a non-empty open set.  Suppose the
functions $R:\RR_+ \times I \to \RR$, $H: \Delta \times I \to \RR$, $b:\RR_+ \times I \to \RR^d$ and
$a: \RR_+ \times I \to \RR^{d \times d}$ are such that $G(t, T, \cdot)$ is twice-continuously 
differentiable for all $(t,T) \in \Delta$ and $G(\cdot, T, z)$ is continuously differentiable for all $(T,z) \in \RR_+ \times I$
and satisfies the partial differential equation
$$
\partial_t G +\sum_{1 \le i \le d} b_i \partial_{z_i} G + \frac{1}{2}\sum_{1 \le i,j \le d} a_{ij}  \partial_{z_i z_j} G = RG
\mbox{ on } \Delta \times I
$$
with boundary conditions
$$
G(T,T,z) = 1 \mbox{ for all } (T,z) \in \RR_+ \times I.
$$
Furthermore, suppose that there exists an interger $n$ and  functions $g_k:\Delta \to \RR$ such that $g(\cdot, T)$
is differentiable for all $T \ge 0$ and 
$$
G(t,T, z) = \sum_{k_1 + \ldots + k_d \le n} g_k(t, T) z^k
$$
and that the functions $(g_k)_k$ are linearly independent.

Then
\noindent Case $n=1$.  For all $t> 0$, the function $R(t, \cdot)$ is a polynomial of degree at most one, for each $i$ the function
$b_i(t, \cdot)$ is a polynomial
of degree at most two, and $a(t, \cdot)$ is unrestricted.

\noindent Case $n\ge 2$. For all $t \ge 0$, the function $R(t, \cdot)$ is a polynomial of degree at most two, for each $i$ the function
$b_i(t, \cdot)$ is a polynomial of degree at most three, and for each $i,j$ the function 
$a_{ij}(t, \cdot)$ is a polynomial of degree at most four. 

Additionally, in the case where $d=1$, if   $R(t,z) = R_0(t) + R_1(t) z+ R_2(t) z^2$, $b(t,z) = b_0(t) + b_1(t) z + b_2(t)z^2+ b_3(t) z^3$ and
$a(z) = a_0(t)+ a_1(t) z + a_2(t) z^2 + a_3(t) z^3 + a_4(t) z^4$,
then the coefficients are such that
$$
 R_2(t) =  \tfrac{n}{2} b_3(t) = - \tfrac{n(n-1)}{2} a_4(t) \mbox{ and } R_1(t) = n b_2(t) + \tfrac{n(n-1)}{2} a_3(t) 
$$
and  $(g_0, \ldots, g_n)$ solves the system of linear ordinary differential equations
\begin{align*}
-\partial_t g_k = & g_{k-2} \left( (k-2) b_3    + \frac{(k-2)(k-3)}{2} a_4 -R_2 \right) \\
&  + g_{k-1} \left( (k-1) b_2 + \frac{(k-1)(k-2)}{2} a_3 - R_1 \right) 
 + g_k \left( kb_1 + \frac{k(k-1)}{2} a_2 - R_0 \right) \\
& + g_{k+1} \left( (k+1)b_0 + \frac{k(k+1)}{2} a_1 \right) +  g_{k+2} \frac{(k+2)(k+1)}{2} a_0 \mbox{ on } [0,T]  \\
g_k(T,T) = & \left\{ \begin{array}{ll} 1 & \mbox{ if } k = 0 \\ 0 & \mbox{ if } k \ge 1, \end{array} \right.
\end{align*} 
  where we interpret $g_{-2} = g_{-1} = g_{n+1} = g_{n+2} = 0$.
\end{theorem}

The proof is essentially the same as that of Theorems \ref{th:main} and \ref{th:main1}, so is omitted.

\subsection{Brody--Hughston rational   model}
In the paper \cite{BH} of Brody \& Hughston, the following rational   model is discussed.  Let $M$ 
be a positive martingale under the \textit{objective measure} $\PP$, and suppose $M_0 = 1$.
Set
$$
V_t = \alpha(t) + \beta(t) M_t
$$
where $\alpha$ and $\beta$ are positive, continuously differentiable, deterministic functions. 
 The idea is that $V$ is a model for the state price density.  Therefore, bond
prices are given by the formula
\begin{align*}
P_{t,T} &= \frac{1}{V_t} \EE^{\PP}(V_T | \ff_t) \\
 &= \frac{\alpha(T) + \beta(T) M_t}{\alpha(t) + \beta(t) M_t}.
\end{align*}
Note that the bond prices are a rational function of the   random variable $M_t$, 
giving the model its name.  Furthermore,
by setting $\alpha(t) + \beta(t) = P_0(t)$ for $t \ge 0$, this model can match the initial term structure of interest rates.

On the other hand, notice that we can write the bond prices as
$$
P_{t,T} = g_0(t,T) + g_1(t,T) Z_t
$$
where the coefficients are defined by
$$
g_0(t,T) = \frac{\beta(T)}{\beta(t)} \mbox{ and } g_1(t,T)  = 
\frac{ \alpha(T) \beta(t)-  \beta(T) \alpha(t))}{\beta(t) }  
$$
and where we let
$$
Z_t = \frac{1}{V_t}
$$
be the factor process.  In particular, this is an affine factor model and hence should
be described by Theorem \ref{th:HW}.  We now carry out the verification
under the assumption that 
$$
dM_t = \nu(t, M_t) M_t dB_t
$$
where $B$ is a $\PP$-Brownian motion and $\nu$ is bounded.

In this framework, we can   define the spot rate as
\begin{align*}
r_t & = - \partial_T P_{t,T}|_{T=t} \\
&=  - \frac{\dot \alpha(t) + \dot \beta(t) M_t}{\alpha(t) + \beta(t) M_t} \\
& = R_0(t) + R_1(t) Z_t
\end{align*}
where
$$
R_0(t) = -\frac{\dot \beta(t)}{\beta(t)}  \mbox{ and }
 R_1(t)  =   \frac{\dot\beta(t) \alpha(t) - \dot \alpha(t) \beta(t)   }{\beta(t) }.
$$
Note that
\begin{align*}
dV_t &= \left( \dot \alpha(t) + \dot \beta(t) M_t \right) dt + \beta(t) dM_t \\
& = - V_t( r_t dt + \lambda_t dB_t)
\end{align*}
where $\lambda_t$ is the market price of risk defined by
\begin{align*}
\lambda_t &=  - \frac{ \beta(t) \nu(t, M_t) M_t }{\alpha(t) + \beta(t) M_t} \\
& =  \nu(t,M_t)(\alpha(t) Z_t - 1)
\end{align*}
Since the process $(\lambda_t)_{0 \le t \le T}$ is bounded, we can define the equivalent risk-neutral pricing measure $\QQ$ by
\begin{align*}
\frac{d \QQ}{d \PP} & =   e^{ \int_t^T r_s ds}  V_t  \\
 & = e^{- \frac{1}{2}\int_0^T \lambda_t dt + \int_0^T \lambda_t dB_t}
\end{align*}
to recover the usual pricing formula
$$
P_{t,T} = \EE^{\QQ}[ e^{- \int_t^T r_s ds} | \ff_t].
$$

Finally, we consider the dynamics of the factor process $Z=V^{-1}$. By It\^o's formula we have
\begin{align*}
dZ_t & = Z_t [ (r_t + \lambda_t^2) dt + \lambda_t dB_t ] \\
& =  (b_1(t) Z_t + b_2(t) Z_t^2) dt + \sigma(t, Z_t) dW_t
\end{align*}
where
\begin{align*}
b_1(t) &=  -\frac{\dot \beta(t)}{\beta(t)} \\
 b_2(t)  &  =   \frac{\dot\beta(t) \alpha(t) - \dot \alpha(t) \beta(t)   }{\beta(t) } \\
\sigma(t,z) & = \nu\left(t, \frac{1- z \alpha(t)}{z \beta(t)} \right)( \alpha(t) z - 1)z
\end{align*}
and where the process $(W_t)_{0 \le t \le T}$ defined by
$$
W_t = B_t + \int_0^t \lambda_s ds
$$
is a $\QQ$ Brownian motion    by Girsanov's theorem. 
In particular, notice that the drift is quadratic in $Z$ and $b_2(t) = R_1(t)$
as predicted by Theorem \ref{th:HW}, while the volatility
is determined by the function $\nu$.

\section{Appendix: Proof of the special Feller test}
Recall that Feller's test is
$$
\PP(T = \infty) = 1 \Leftrightarrow  v(z_{\mathrm{min}}) = \infty = v(z_{\mathrm{max}}),
$$
where   Feller's test function is defined by
	$$
v(x)  = 
\int_{z=c}^x \int_{y=z}^x \frac{1}{a(z)} e^{\int_y^z \frac{2b(w)}{a(w)}dw} dy \  dz, \mbox{ for }  z_{\mathrm{min}} 
< x < z_{\mathrm{max}},
$$
where $c = \frac{1}{2}( z_{\mathrm{min}} +  z_{\mathrm{max}})$. 
See for instance Chapter 5 of  Karatzas and Shreve's \cite{KS} book.

It is enough to consider the behaviour of $v$ near $x =  z_{\mathrm{min}}$, as the behaviour
near $x=z_{\mathrm{max}}$ is analogous.

By changing variables, we now study the cases where the integral
$$
v(z_{\mathrm{min}}) = \int_{ z_{\mathrm{min}}}^c \frac{p(z)}{a(z) p'(z)}  dz
$$
is finite or infinite,
where
$$
p(z) =  \int_{z_{\mathrm{min}}}^z e^{\int_y^c \frac{2b(u)}{a(u)}du} dy.
$$
is the related to the scale function.
Now, by assumption the functions $a$ and $b$ are polynomials, and hence
near $z_{\mathrm{min}}$ can be written as
\begin{align*}
a(t+z_{\mathrm{min}}) &= \alpha t^{A+1} + O(t^{A+2} ) \\
b(t+ z_{\mathrm{min}}) & = \beta t^{B} + O(t^{B+1})
\end{align*}
for constants $\alpha> 0$ and  $\beta \ne 0$ and for integers $A, B \ge 0$.
Note that with this notation
$$
2b(z_{\mathrm{min}}) - a'(z_{\mathrm{min}}) = \beta \1_{\{ B=0 \}} - \alpha \1_{\{A = 0\} }.
$$
Hence, we must show that $v(z_{\mathrm{min}}) = \infty$ on
$$
\{ A > 0, B = 0, \beta > 0 \} \cup \{ A > 0, B > 0 \} \cup \{A =0, B=0, 2\beta \ge \alpha \}
$$
and that $v(0) < \infty$ on the complement
$$
\{  A > 0, B= 0,  \beta < 0 \} \cup \{ A = 0, B > 0 \} \cup \{A =0, B=0, 2\beta < \alpha \}.
$$

We have the calculation
$$
 \int_{t+z_{\mathrm{min}}}^c \frac{2b(s)}{a(s)}ds  =\left\{ \begin{array}{ll}  \mathrm{const} + O(t) & \mbox{ if } B \ge A+1 \\ 
- \frac{2\beta}{  \alpha} \log t + \mathrm{const} + O(t) & \mbox{ if } B = A \\ 
   \frac{2 \beta}{  \alpha (A-B) } t^{-(A-B) } + O( t^{1-A+B } ) & \mbox{ if } B \le  A -1 \\ 
\end{array}  \right.
$$
and hence
$$
 p'(t + z_{\mathrm{min}})  =\left\{ \begin{array}{ll}  \mathrm{const}(1 + O(t)) & \mbox{ if } B \ge A+1 \\ 
t^{- 2 \beta/\alpha}( \mathrm{const}  + O(t) ) & \mbox{ if } B = A \\ 
   e^{ \frac{2 \beta}{  \alpha (A-B) } t^{-(A-B) }} (1 + O(t) )  & \mbox{ if } B \le  A -1 \\ 
\end{array}  \right.
$$
and therefore
$$
 \frac{p(t + z_{\mathrm{min}})}{a(t + z_{\mathrm{min}}) p'(t + z_{\mathrm{min}})}
  =\left\{ \begin{array}{ll}  \frac{1}{\alpha}t^{-A}(1 + O(t)) & \mbox{ if } B \ge A+1 \\ 
\infty & \mbox{ if } B = A, 2 \beta \ge \alpha \\ 
\frac{1}{2 \beta} t^{-A}(1  + O(t) ) & \mbox{ if } B = A, 2 \beta < \alpha \\ 
\infty & \mbox{ if } B \le A - 1,   \beta > 0 \\ 
  \frac{1}{2 \beta}t^{-B}  (1 + O(t) )  & \mbox{ if } B \le  A -1, \beta < 0.\\ 
\end{array}  \right.
$$
From this, we see that $v(z_{\mathrm{min}}) = \infty$ precisely on 
\begin{align*}
\{ B \ge A+1, A \ge 1 \} &\cup \{ B = A, 2 \beta \ge \alpha \} \cup 
 \{ B = A \ge 1, 2 \beta < \alpha \} \\
& \cup \{ A \ge B+1, \beta > 0 \} \cup
\{ A \ge B+1 \ge 2, \beta < 0 \}
\end{align*}
from which the conclusion follows.

\section{Acknowledgement}
The SC acknowledges the financial support of the Man Group studentship and MRT that of the Cambridge Endowment for Research in Finance.
This work has been presented at the Labex Louis Bachelier - SIAM-SMAI Conference on Financial Mathematics in Paris,
the Conference on Stochastic Calculus, Martingales and Financial Modeling in St Petersburg,
the Advanced Methods in Mathematical Finance in  Angers, and the Eleventh Cambridge--Princeton Conference
in Princeton.  We would like to thank the  participants for useful comments and suggestions.
Finally, we would like to thank Dorje Brody and Lane Hughston for a discussion of this work
and its relation to their paper \cite{BH}.

\end{document}